\documentclass{article}
\usepackage[margin=1in,footskip=0.25in]{geometry}
\usepackage{spconf, amsfonts, amsmath}
\usepackage{xspace}
\usepackage{algorithm}
\usepackage{algpseudocode}
\usepackage{bm}
\usepackage{subcaption}
\usepackage{graphicx}
\usepackage{amssymb}
\usepackage{amsthm}

\usepackage{hyperref}

\hypersetup{
    colorlinks, 
    citecolor=black,
    filecolor=black,
    linkcolor=black,
    urlcolor=blue
} 

\newcommand{\xv}{{\bf x}}
\newcommand{\Real}{\mathbb{R}}
\newcommand{\ie}{i.e.,\xspace}
\newcommand{\sv}{{\bf s}}
\newcommand{\tv}{{\bf t}}
\newtheorem{theorem}{Theorem} 
\begin{document}

\setlength\abovedisplayskip{3pt}
\setlength\belowdisplayskip{3pt}
\title{Low-complexity vector source coding for discrete long sequences with unknown distributions}
\name{
Leah Woldemariam, Hang Liu and Anna Scaglione
	\thanks{This work was supported in part by the DoD-ARO under Grant No. W911NF2210228. This material is based upon work supported by the National Science Foundation Graduate Research Fellowship under Grant No. DGE – 2139899.}
}
\address{
	Department of Electrical and Computer Engineering, Cornell Tech, Cornell University, NY,  USA\\
 	Emails: \{lsw85,hl2382,as337\}@cornell.edu
}

\maketitle
\begin{abstract} 

In this paper, we propose a source coding scheme that represents data from unknown distributions through frequency and support information. Existing encoding schemes often compress data by sacrificing computational efficiency or by assuming the data follows a known distribution. We take advantage of the structure that arises within the spatial representation and utilize it to encode run-lengths within this representation using Golomb coding. Through theoretical analysis, we show that our scheme yields an overall bit rate that nears entropy without a computationally complex encoding algorithm and verify these results through numerical experiments. 
\end{abstract}
\begin{keywords}
Data compression, run-length coding, vector coding, bit analysis
\end{keywords}

\section{Introduction}
The increase of data generated over large networks and the idea of leveraging edge intelligence for distributed learning  has renewed research interest in efficient data storage and communication. 
In these scenarios, the statistical attributes of data or learning models that are exchanged through the network are often not known \emph{a priori}. Typical compression methods for a large volume of data include quantization or discretization
, sparsification \cite{wen2017terngrad, shi2019sparsification, sattler2019sparse}, and source coding \cite{yang2023federated, zhang2021adaptive}
. This work focuses on designing efficient lossless source coding for lengthy sequences of discrete symbols with an unknown underlying distribution.

Shannon notably established in \cite{shannon1948mathematical} that the optimal compression ratio achievable through lossless coding is fundamentally constrained by the entropy of the input distribution. Attaining Shannon's lower bound has attracted constant research attention, especially when catering to diverse data models. For instance, Golomb coding with an optimized parameter is known to reach a bit length close to the entropy for geometric distributions \cite{golomb1966run}. For more general inputs, Huffman coding is considered optimal with integer bit lengths \cite{huffman1952method}. 
Moreover, arithmetic coding, which uses fractional bit lengths, draws closer to the entropy limit by relying on the precision of fraction representations of bit strings \cite{rissanen1976generalized, pasco1977source}. 

However, entropy coding techniques often entail significant computational complexity in parameter tuning or intricate codebook construction \cite{kiely2004selecting, huffman1952method}.

Run-length (RL) coding offers a more straightforward, albeit sub-optimal, approach to source coding. It divides the input sequence into several sub-sequences, or `runs,' characterized by consecutive repetitions of the same value and encodes the length of each run. 
It is particularly effective in coding sequences with biased distributions that lead to long runs. 
Furthermore, RL coding can be seamlessly combined with other integer coding methods to enhance the compression ratio. A recent study in \cite{gajjala2020huffman} proposed to combine RL coding with Huffman coding for model compression in deep learning. Similarly, \cite{alistarh2017qsgd} incorporated RL coding with Elias coding for the communication of quantized models in distributed machine learning. However, the efficacy of RL coding critically depends on the chosen coding schemes for run lengths. Ideally, these should be determined with insights from the knowledge or estimation of input probability models. To our understanding, there is no universal RL coding scheme that reaches the entropy bit-length limit for any unknown distribution. 

In this paper, we propose a novel encoding scheme for discrete input sequences of independent and identically distributed (i.i.d) symbols. Our scheme encodes frequency and positional information without needing expensive encoding and decoding algorithms. We also provide an asymptotic analysis of the bit rate associated with our scheme proving our scheme yields a near-entropy bit rate.

\section{Problem Formulation}
In this work, we consider the source coding problem for the lossless compression of a high-dimensional discrete input vector $\xv=[x_1,x_2,...,x_N]$, where $N$ denotes the input length. Each $x_n$ takes its value from a finite alphabet $\mathcal{Q}=\{q_0,q_1,\cdots,q_{L\!-\!1}\}$, where $q_\ell\in\Real$ is the $\ell$-th possible value, and $L$ denotes the range of possible values. We assume that each element $x_i$ is i.i.d. with some \emph{unknown} probability mass function (PMF), denoted by $f$.

We explore a scenario where the input length significantly exceeds the number of potential input values, \ie $N\gg L$. This configuration is commonly encountered in many fields, such as image and model compression for deep learning. Moreover, we assume that both the encoder and decoder are aware of the alphabet $\mathcal{Q}$.\footnote{If $\mathcal{Q}$ remains unknown to the decoder, the encoder can send the $L$ values in $\mathcal{Q}$ to the decoder with high precision, incurring a communication overhead of $\mathcal{O}(L)$. Given $N\gg L$, this overhead is negligible compared to that of transmitting $\xv$.} We note that each $q_{\ell} \in \mathcal{Q}$ can be distinctly represented by its respective index in $\mathcal{Q}$. Therefore, encoding $\xv$ is equivalent to the \emph{integer} encoding of a vector of indices representing its elements, and we can assume the input alphabet to be $\mathcal{Q}=\{0,1,\cdots,L-1\}$. 

The goal of this work is to develop a lossless encoding scheme $\mathcal{E}(\cdot)$ that maps $\xv$ to a binary code, as well as a corresponding decoding scheme $\mathcal{D}(\cdot)$ reconstructing $\xv$ from the code received with no error. Our aim is minimizing the length of the encoded bits sequence $|\mathcal{E}(\xv)|$, regardless of the unknown input distribution. As shown in Shannon's source coding theorem, the minimal bit length per symbol of any encoder is bounded below by the entropy $H(f)$ \cite{shannon1948mathematical}. 
Existing techniques that approach this entropy limit, such as Golomb and Huffman coding, are optimal for certain distributions \cite{golomb1966run, huffman1952method}, but efficient implementations of these methods require careful hyper-parameter fine-tuning and may suffer from high computational complexity \cite{kiely2004selecting, huffman1952method}. In contrast, we propose a tuning-free near-entropy coding scheme $\mathcal{E}$ tailored for any arbitrary unknown distribution in the next section.

\section{Proposed Method}
Our method is built upon the observation that the input $\xv$ can be described by the frequencies and positions of its possible values within $\mathcal{Q}$. To elucidate, we define the type of $\xv$ as $ {\bf t(\xv)}=[t_0(\xv),\ldots,t_{L-1}(\xv)] \in\mathbb{Z}^L$. Here, $t_\ell(\xv)$ represents the frequency of $\ell$ in $\xv$, \ie
\begin{align} \label{eq01}
t_{\ell}(\xv)=\sum_{n=1}^N \delta(x_n-\ell),
\end{align}
where $\delta(\cdot)$ is the Dirac delta function satisfying $\delta(x)=1$ when $x=0$ and $\delta(x)=0$ otherwise. For every $0\leq\ell\leq L-1$, we define an $N$-dimensional support vector $\sv_\ell\in\{0,1\}^N$ with the $n$-th entry given by 
\begin{align} \label{eq02}
    s_\ell[i]= \delta(x_n-\ell), ~~\text{}~~ 1\leq n\leq N.
\end{align} 
Consequently, the values of $\xv$ can be succinctly expressed through $\tv(\xv)$ and $\{\sv_\ell\}_{\ell=0}^{L-1}$. We propose a scheme $\mathcal{E}(\xv)$ to encode $\tv(\xv)$ and $\{\sv_\ell\}_{\ell=0}^{L-1}$ with the following steps:

\noindent\textbf{Type encoding}. We consecutively encode the $L$ elements of $\tv(\xv)$ in \eqref{eq01}. Given that the types are large integers, we adopt Elias omega coding from \cite{elias1975universal} for the encoding, which is recognized as a universal code for positive numbers.

\noindent\textbf{RL representation of the support vectors}. We propose a representation of the $L$ support vectors $\{\sv_\ell\}_{\ell=0}^{L-1}$ that reduces their dimensionality. To begin, we rearrange the type vector $\tv(\xv)$ in a descending order. Note that this order is also clear to the decoder upon the decoding of the type vector. 
For simplicity of notation, we assume without loss of generality that the entries in $\tv(\xv)$ are already arranged such that $t_0(\xv)\geq t_1(\xv)\geq \cdots \geq t_{L-1}(\xv)$.

For any symbol index $1\leq n\leq N$, the set $\{s_{\ell}[n]\}_{\ell=0}^{L-1}$ contains $L-1$ zeros with only one entry being one. Thus,
\begin{align}\label{eq03}
\sum_{\ell=0}^{L-1}s_{\ell}[n]=1, \forall n.
\end{align}
This observation suggests that encoding only the last $L-1$ support vectors is sufficient as the first support vector $\sv_0$ can be deduced by $s_0[n]=1-\sum_{\ell=1}^{L-1}s_\ell[n],\forall n.$

Furthermore, for any $\ell\geq 1$ and $\forall n$, $s_\ell[n]$ is zero when there exists an $\ell^\prime<\ell$ such that $s_{\ell^\prime}[n]=1$.
Motivated by this fact, we represent $\sv_\ell$ for $\ell=1,\cdots,L-1$ with iteratively shrinking subsets. We define the index sets $\mathcal{I}_\ell,1\leq \ell \leq L-1$, with $\mathcal{I}_1=\{1,\cdots,N\}$ and $\mathcal{I}_{\ell+1}=\mathcal{I}_\ell\setminus\{1\leq n\leq N:s_\ell[n]=1\}$, where $\setminus$ is the set difference operation.
Let $\sv_\ell[\mathcal{I}_\ell]$ denote the sub-vector of $\sv_\ell$ indexed by $\mathcal{I}_\ell$.
It is sufficient to use $\{\sv_\ell[\mathcal{I}_\ell]\}_{\ell=1}^{L-1}$ to represent $\{\sv_\ell\}_{\ell=0}^{L-1}$.
Finally, we represent each $\sv_\ell[\mathcal{I}_\ell]$ for $\ell=1,\cdots,L-1$, using its RLs of zeros, defined as

the numbers of consecutive zeros between two ones. For each $\ell$, we only need to encode the first $t_\ell(\xv)$ RLs and drop the last one as the sum of all of the RLs is equal to $|\mathcal{I}_\ell|-t_\ell(\xv)$. For example, the RLs of $[0,0,1,0,1,0,0,0]$ are $[2,1,3]$, and we shall encode $2$ and $1$.

\noindent\textbf{Support encoding}. For $\ell\geq 1$, we iteratively encode the RLs of $\sv_\ell[\mathcal{I}_\ell]$ by using Golomb coding \cite{golomb1966run}. This method requires a predefined hyper-parameter $M$ for the separate encoding of an integer $X$ based on its quotient and its remainder. 
The quotient $\lfloor X/M\rfloor$ is encoded by unary coding, followed by the truncated binary code for the remainder $X-\lfloor X/M\rfloor$.

For each iteration $\ell$, the Golomb coding parameter $M_\ell$ is chosen to minimize the code length, and has an optimal value given by the nearest integer value of $(\ln 2) (N-\sum_{m\leq \ell} t_m(\xv))/t_\ell(\xv)$, which we denote by $[(\ln 2) (N-\sum_{m\leq \ell} t_m(\xv))/t_\ell(\xv)]$. The derivation of this value can be found in Section \ref{sec4}. Finally, we sequentially append the Golomb codes of the RLs to the codes of $\tv(\xv)$.\footnote{The value of each $M_\ell$ can be computed using $\tv(\xv)$ at the decoder side. Therefore, there is no need to additionally encode $M_\ell$.}

We summarize the encoding scheme in Algorithm \ref{alg:cap}. The decoding scheme is presented as follows.

\noindent\textbf{Decoding}. The decoder begins by decoding the first $L$ integers using the Elias omega decoder to reconstruct the type vector $\tv(\xv)$. Next, all the remaining bits are decoded to extract the RLs using the corresponding Golomb decoding with parameter $M_\ell$.
Note that each $\mathcal{I}_{\ell+1},\ell\geq 1,$ can be computed by the values of $\mathcal{I}_{\ell}$ and $\sv_\ell[\mathcal{I}_\ell]$. Starting from iteration $\ell=1$ to $L-1$, we 
retrieve $\mathcal{I}_\ell$ and $\sv_\ell[\mathcal{I}_\ell]$ from the decoded RLs. The original support vectors $\{\sv_\ell\}_{\ell=1}^{L-1}$ are obtained by zero-padding $\{\sv_\ell[\mathcal{I}_\ell]\}_{\ell=1}^{L-1}$ at the removed positions. 
Then, $\sv_0$ is computed and $\xv$ is recovered by \eqref{eq03} and \eqref{eq02}, respectively.

\begin{algorithm}[hbt!] 
\caption{The proposed encoding scheme. \label{alg:cap}}
\begin{algorithmic}
\Require $\xv \in \mathbb{R}^N$, $L > 0$.
\State Compute $\bm t(\bm x)$ and encode it with Elias omega encoding.
\State Initialize $\mathcal{I} = \{ 1, \dots, N\}$.
\For{$\ell = 1, \dots, L-1$}
    \State $\mathcal{I}^\prime = \emptyset$, $\mathcal{R}_{\ell} = \emptyset$.
    \State Compute $M_{\ell} = [(\ln 2) (N-\sum_{m\leq \ell} t_m(\xv))/t_\ell(\xv)]$.
    \State $r = 0$.
    \For{$n \in \mathcal{I}$}
    \If{$\delta(x_n - \ell) = 0$}
    \State $\mathcal{R}_{\ell} = \mathcal{R}_{\ell} \cup \{ r \}$.
    \If{$|\mathcal{R}_{\ell}| < t_{\ell}(\bm x)$}
        \State Encode $r$ with Golomb coding.
    \EndIf
    \State $r = 0$.
    \State $\mathcal{I}^\prime = \mathcal{I}^\prime \cup \{n\}$.
    \Else
    \State $r = r + 1$.
    \EndIf
    \State $\mathcal{I} = \mathcal{I}$ $ \backslash$ $\mathcal{I}^\prime$.
    \EndFor
\EndFor
\end{algorithmic}
\end{algorithm}

\begin{figure*}[h]
    \centering
    \begin{subfigure}[b]{0.33\textwidth}
        \includegraphics[width=\textwidth]{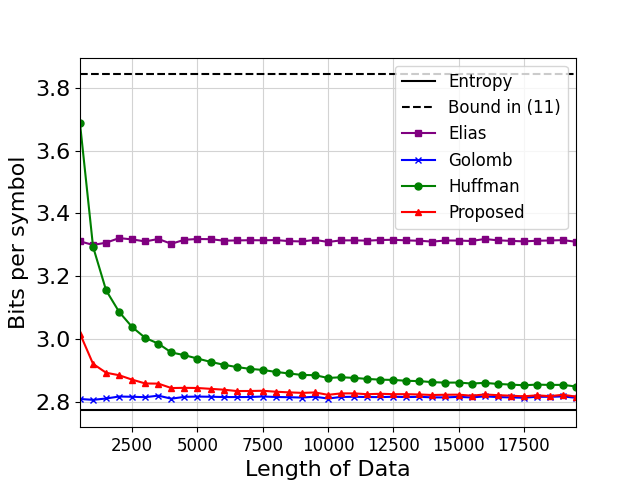}
        \caption{}
        \label{fig:geo}
    \end{subfigure}
    \begin{subfigure}[b]{0.33\textwidth}
        \includegraphics[width=\textwidth]{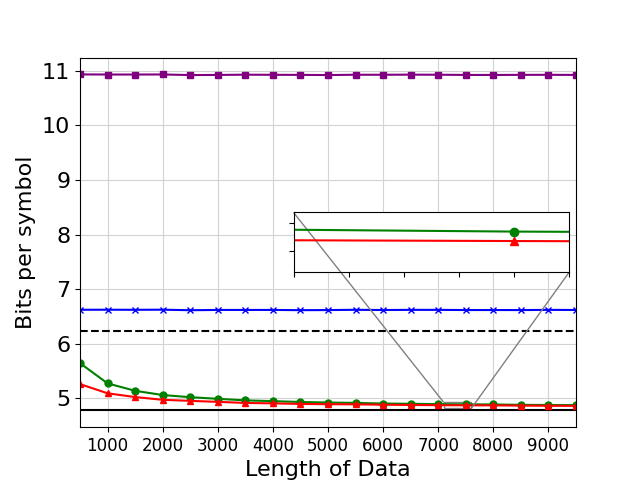}
        \caption{}
        \label{fig:bim}
    \end{subfigure}
    \begin{subfigure}[b]{0.33\textwidth}
        \includegraphics[width=\textwidth]{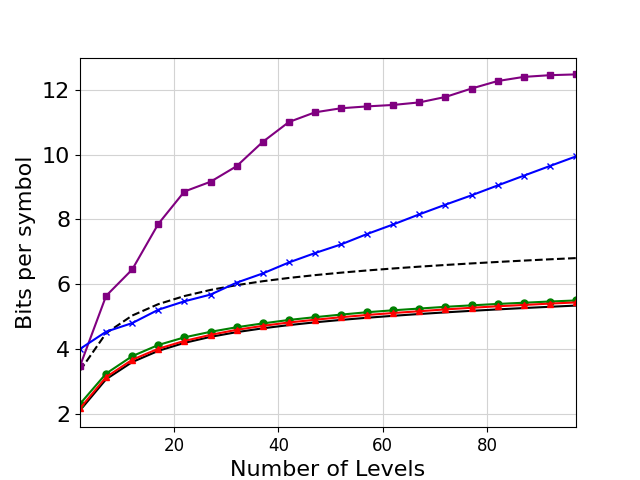}
        \caption{}
        \label{fig:fixN}
    \end{subfigure}
    \label{fig:foobar}
    \caption{(a) Bit length per symbol versus the data length $N$ with inputs i.i.d. drawn from the geometric distribution with parameter equal to $0.33$ and with $L=50$. (b) Bit length per symbol versus the data length $N$ for the bimodal distribution. (c) Data compression performance for the bimodal distribution under a varying value of $L$, where the input length $N=10,000$.}
\end{figure*}

\section{Asymptotic Analysis}\label{sec4}
In this section, we analyze the code length and the computational complexity of the proposed method in the large-system limit with $N\to \infty$ and $L=o(N)$.

\noindent\textbf{Bit analysis}. For $\xv$ with entries i.i.d. drawn from PMF $f$, define $f_\ell\triangleq f(x_i=\ell)$ for $0\leq \ell\leq L-1$. Without loss of generality, we assume in this section that $f_0\geq f_1\geq\cdots\geq f_{L-1}$. We aim to derive an upper bound for the average bit length per symbol, denoted by $\frac{1}{N}{|\mathcal{E}(\xv)|}$, as $N\to \infty$.

As $\sum_{\ell=0}^{L-1}t_\ell(\xv)=N$, it is shown in \cite[Lemma A.3]{alistarh2017qsgd} that the code length per symbol for encoding $\tv(\xv)$ with the Elias omega coding is upper bounded by
\begin{align}\label{eq04}
\frac{L}{N}\left((1+o(1))\log_2 \frac N L+1\right),
\end{align}
where $o(1)$ is a diminishing term as $N\to \infty$. It follows from \eqref{eq04} that the bit length per symbol for encoding the types tends to zero as $N\to \infty$. The next theorem characterizes the bit length for encoding each support vector $\sv_\ell[\mathcal{I}_\ell]$.
\newcommand{\round}[1]{\text{round}\left(#1\right)}
\begin{theorem}\label{the1}For any $1\leq \ell\leq L-1$ and $N\to \infty$, encoding the RLs of $\sv_\ell[\mathcal{I}_\ell]$ using the proposed method yields a bit length no larger than
\begin{align}\label{eq05}
f_\ell\left(\log_2\frac{1-\sum_{m=1}^{\ell-1} f_m(\xv)}{f_\ell(\xv)}+2.914\right).
\end{align}
In particular, the optimal choice of the Golomb parameter $M_\ell$ that leads to \eqref{eq05} is given by the nearest integer value of $(\ln 2) (1-\sum_{1\leq m\leq \ell} f_m(\xv))/f_\ell(\xv)$.
\end{theorem}
\begin{proof}
Let $\mathcal{R}_\ell$ denotes the corresponding RL set of $\sv_\ell[\mathcal{I}_\ell]$. It follows from the definition of the RL that

\begin{align}
&|\mathcal{R}_\ell|=t_\ell(\xv),\label{eq06}\\
&\sum_{r\in \mathcal{R}_\ell}r \leq \sum_{i\in \mathcal{I}_\ell} \delta(s_\ell[i])=|\mathcal{I}_\ell|-t_\ell(\xv)=N-\sum_{m=1}^{\ell} t_m(\xv).\label{eq07}
\end{align}
By the law of large numbers, the types in $\tv(\xv)$ converge to their expected values as $N\to \infty$, \ie
\begin{align}\label{eq08}
\lim_{N\to \infty} \frac{t_{\ell}(\xv) }{N} = \mathbb{E}_{x\sim f}[t_{\ell}({\bf x})] = f_{\ell},\forall \ell.
\end{align}

As $N\to \infty$, the number of bits per symbol in encoding $\mathcal{R}_\ell$ with the Golomb parameter $M_\ell$ is bounded by
\begin{align}
&\lim_{N\to \infty}\frac{1}{N}\sum_{r\in \mathcal{R}_\ell} \left(\overbrace{\left\lfloor \frac{r}{M_\ell}\right\rfloor +1}^{\text{Quotient encoding}} + \overbrace{\left\lfloor \log_2 M_\ell\right\rfloor+1}^{\text{Remainder encoding}}\right) \nonumber\\
\leq &\lim_{N\to \infty}\frac{1}{N}\left( \frac{1}{M_\ell}\sum_{r\in \mathcal{R}_\ell} r+|\mathcal{R}_\ell|\left(\log_2 M_\ell+2\right)\right)\nonumber\\
\overset{\eqref{eq06},\eqref{eq07}}{\leq}&\lim_{N\to \infty}\frac{N-\sum_{m=1}^{\ell} t_m(\xv)}{NM_\ell}+\frac{t_{\ell}({\bf x})}{N}(\log_2M_\ell + 2)\nonumber\\
\overset{\eqref{eq08}}{=}&\frac{1-\sum_{m=1}^{\ell} f_m}{M_\ell}+f_\ell(\log_2M_\ell + 2)\triangleq g(M_\ell). 
\label{eq09}
\end{align}
Using the first-order optimality condition of one-dimensional functions, one can show from the derivative of $g(M_\ell)$ that the optimal $M_\ell\in \mathbb{Z}$ that minimizes $g(M_\ell)$ is given by $(\ln 2) (1-\sum_{1\leq m\leq \ell} f_m(\xv))/f_\ell(\xv)$. The minimum value of $g(M_\ell)$ is further bounded by
\begin{align}
\min_{M_\ell} g(M_\ell)\leq& g\left(\left\lceil\ln 2 \frac{1-\sum_{m=1}^{\ell} f_m(\xv)}{f_\ell(\xv)}\right\rceil\right)\nonumber\\
< &f_\ell\left(\log_2\frac{1-\sum_{m=1}^{\ell-1} f_m(\xv)}{f_\ell(\xv)}+2.914\right).
\end{align}
\end{proof}

From Theorem \ref{the1}, the overall asymptotic code length of the proposed method satisfies the following inequality:
\begin{align}
&\lim_{N\to \infty}\frac{1}{N}|\mathcal{E}(\xv)|< \sum_{\ell=1}^{L-1} f_\ell\left(2.914+\log_2 \left(\frac{1-\sum_{m=1}^{\ell-1}f_m}{f_\ell}\right)\right)\nonumber\\
=&2.914(1-f_0)-\sum_{\ell=1}^{L-1} f_\ell \log_2 f_\ell+\sum_{\ell=1}^{L-1}f_\ell \log_2(1-\sum_{m=1}^{\ell-1}f_m)\nonumber\\
=&H(f)+C(f),\label{analytical}
\end{align}\\
where $C(f)=2.914(1-f_0)+f_0\log_2 f_0+\sum_{\ell=1}^{L-1}f_\ell \log_2(1-\sum_{m=1}^{\ell-1}f_m)$ is the maximum gap between the code length of the proposed method and the entropy lower bound.

In practice, the choice of $M_{\ell}$ used in Algorithm \ref{alg:cap} is only an approximation to the exact optimal parameter as the true underlying PMF $f$ is unknown. By \eqref{eq08}, this choice of $M_{\ell}$ converges to the optimal parameter.

\noindent
\textbf{Complexity analysis.} 
Next, we compare the computational complexity of our proposed method with the existing baseline. Our proposed scheme has a  complexity of $\mathcal{O}(N\log L)$ to encode and decode the RLs, as the complexity of encoding each run is $\mathcal{O}(\log L)$ with a total of $\sum_\ell t_{\ell}({\bf x}) = N$ runs to encode; see \eqref{eq09}. In addition, encoding the type vector has a complexity of $\mathcal{O}(L\log L)$. For comparison, Huffman coding requires creating a length-$L$ codebook based on the frequencies of all the possible values. Then each symbol is encoded by searching the codebook, leading to a total worst-case computational complexity of $\mathcal{O}(NL)$. In comparison, our proposed method has an overall lower complexity than Huffman coding with large $N$.

\section{Numerical Results}
We evaluate the performance of the proposed scheme via simulations. We compare our method in terms of coding bit length with baselines, including symbol-wise Golomb coding, symbol-wise Elias coding, and Huffman coding, as well as the entropy lower bound. 

Fig. \ref{fig:geo} plots the average bit length per symbol over 50 Monte Carlo trials on the i.i.d. geometric distribution with parameter 0.33.
In this setup, Golomb coding is deemed to attain the entropy limit with an optimal selection of the parameter $M$. Through trial and error, we find that $M=2$ offers the best performance. In contrast, our method achieves a comparable compression rate with Huffman coding, which is known to be the optimal character coding scheme. In particular, the bit length of the proposed method converges to the value matching the same integer value of the entropy as $N$ grows. This observation confirms the analysis in Section \ref{sec4}.

In fig. \ref{fig:bim}, we evaluate the bit lengths of different coding schemes using symbols drawn i.i.d. from a synthetic bimodal distribution. Specifically, the distribution is generated by overlaying the binomial distribution with parameters of $50$ and $0.33$ and its duplicate shifted by 20. The overlaid distribution is then truncated to fall within the range of $[0,50]$. 
Similar to Fig. \ref{fig:geo}, Fig. \ref{fig:bim} shows that the code length per symbol of our method converges to the entropy limit as $N$ increases, demonstrating the efficiency of our scheme under different distributions. Fig. \ref{fig:fixN} plots the coding performance in relation to a varying input alphabet range $L$ under the bimodal distribution. The entropy of the input distribution grows with $L$, and eventually converges with a large $L$, because of the exponentially diminishing PMF trend.

\section{Conclusion}
In this work, we studied the source coding of long sequences with i.i.d symbols drawn from an unknown PMF and introduced an efficient support encoding scheme. 
We represented these sequences by their frequency and positional information and iteratively encode the support vectors while discarding the positions encoded in earlier support vectors, whose RLs are encoded with Golomb coding.  
The proposed method avoids the need for codebook construction and parameter tuning, thus leading to a lower computational complexity compared to existing optimal coding schemes. Theoretic analyses show our method asymptotically achieves an average bit length per symbol close to the entropy limit. Finally, numerical results verify the efficiency of our method in various setups.

\vfill\pagebreak

\bibliographystyle{IEEEtran}
\bibliography{ref}

\end{document}